\documentclass[a4paper,12pt]{article}
\usepackage{amsmath,amsthm,amsfonts,amssymb,color}
\usepackage[matrix,arrow]{xy}
\usepackage{graphicx}
\usepackage[UKenglish]{babel}
\usepackage{multicol} 
\usepackage{appendix}
\usepackage{epstopdf}

\usepackage{caption}
\usepackage{subcaption}

\usepackage{multirow} % para las tablas
\usepackage{float}
\usepackage{graphicx}
\usepackage{longtable}  

\makeatletter
\renewcommand\@biblabel[1]{#1.}
\makeatother

\usepackage{amsmath,amssymb,amsthm}
\usepackage{enumerate}
\usepackage{amsfonts}
\usepackage{graphics,graphicx}
\usepackage{verbatim}
\usepackage{latexsym,makeidx}
\usepackage{amsmath,amscd}
\usepackage{appendix}
\usepackage[latin1]{inputenc}

\newenvironment{keywords}{
  \vspace{2mm}
  \noindent
  \keywordsname: 
  \itshape\small
}

\def\keywordsname{\textbf{Keywords}}
  \def\mathsubclassname{\textbf{2010 AMS Subject Classification}}

\newcommand\be{\begin{equation}}
\newcommand\ee{\end{equation}}

 \newtheorem{thm}{Theorem}[section]
 \newtheorem{corollary}[thm]{Corollary}
 \newtheorem{lemma}[thm]{Lemma}
 
 \theoremstyle{definition}
 
 \theoremstyle{remark}
 \newtheorem{remark}[thm]{Remark}
 
 \numberwithin{equation}{section}

\begin{document}

\title{Exact solution for a two-phase Stefan problem with variable latent heat and a convective boundary condition at the fixed face.}

\author{
Julieta Bollati$^{1}$,  Domingo A. Tarzia $^{1}$\\ \\
\small {{$^1$} Depto. Matem\'atica - CONICET, FCE, Univ. Austral, Paraguay 1950} \\  
\small {S2000FZF Rosario, Argentina.}\\
\small{Email: JBollati@austral.edu.ar; DTarzia@austral.edu.ar.} 
}
\date{}

\maketitle

%\subclass{35R35, 80A22, 35C05.}


\begin{abstract}

Recently it was obtained in [Tarzia, Thermal Sci. 21A (2017) 1-11] for the classical two-phase Lam\'e-Clapeyron-Stefan problem an equivalence between the temperature and convective boundary conditions at the fixed face under a certain restriction. Motivated by this article we study the two-phase Stefan problem for a semi-infinite material with a latent heat defined as a power function of the position and a convective boundary condition at the fixed face. An exact solution is constructed using Kummer functions in case that an inequality for the convective transfer coefficient is satisfied generalizing  recent works for the corresponding one-phase free boundary problem. We also consider the limit to our problem when that coefficient goes to infinity obtaining a new free boundary problem, which has been recently studied in [Zhou-Shi-Zhou, J. Engng. Math. (2017) DOI 10.1007/s10665-017-9921-y].


\end{abstract}

\maketitle
\section{Introduction}

The study of heat transfer problems with phase-change such as melting and freezing have attracted growing attention in the last decades due to their wide range of engineering and industrial applications.   Stefan problems can be modelled as basic phase-change processes where the location of the interface is a priori unknown. They  arise in a broad variety of fields like melting, freezing, drying, friction, lubrication, combustion, finance, molecular diffusion, metallurgy and crystal  growth. Due to their importance, they have been largely studied since the last century \cite{AlSo}, \cite{Ca}-\cite{Gu}, \cite{Lu}, \cite{Ru}, \cite{Ta4}. For an account of the theory we refer the reader to \cite{Ta2}.

In the classical formulation of Stefan problems, there are many assumptions on the physical factors involved that are taken into account  in order to simplify the description of the process. 
The latent heat, which is the energy required to accomplish the phase change, is usually considered  constant. However in many practical  problems a constant latent heat may be not appropriate, being necessary     to assume a variable one. The physical bases of this particular assumption can be found in the movement of a shoreline \cite{VSP}, in the ocean delta deformation \cite{LoVo} or in the cooling body of a magma \cite{Per}.

In   \cite{Pr}  sufficient conditions for the existence and uniqueness of solution of a one-phase Stefan problem taking a latent heat as a general function of the position were found.  
In \cite{VSP}, as well as  in \cite{SaTa}  an exact solution was found for a one-phase and two-phase Stefan problem respectively  considering the latent heat as a linear function of the position. \cite{ZWB} generalized \cite{VSP} by considering the one-phase Stefan problem with the latent heat as a power function of the position with an integer exponent. 
Recently in \cite{ZhXi}  the latter problem was studied assuming a real non-negative exponent. It was presented the explicit solution for two different problems 
defined according to the boundary conditions considered: temperature and flux.

Boundary conditions imposed at a surface of a body in order to have a well-posed  mathematical problem, can be specified in terms of temperature or energy flow. One of the most realistic boundary conditions is the convective one, in which the heat flux depends not only on the ambient conditions  but also on the temperature of the surface itself. In \cite{Ta1} it was studied the relationship between a classical two-phase Stefan problem considering temperature and convective boundary condition at the fixed face $x=0$. In \cite{BrNa}, a nonlinear one-phase Stefan problem with a convective boundary condition in Storm's materials was studied.

Motivated by \cite{Ta1} and \cite{ZhXi}, in \cite{BoTa} we studied the one-phase Stefan problem considering a variable latent heat and a convective boundary condition at the fixed face $x=0$.
In the present paper we are going to analyse the existence and uniqueness of solution of a two-phase Stefan problem, considering an homogeneous semi-infinite material, with a latent heat as a power function of the position and a convective boundary condition at the fixed face $x=0$. This problem can be formulated in the following way: find the temperatures $\Psi_l(x,t)$, $\Psi_s(x,t)$ and the moving melt interface $s(t)$ such that:
\begin{eqnarray}
& &{\Psi_{l}}_t(x,t)=d_l {\Psi_l}_{xx}(x,t), \qquad 0<x<s(t), \quad t>0,\qquad\qquad  \label{1}\\
& &{\Psi_s}_t(x,t)=d_s {\Psi_s}_{xx}(x,t),\qquad x>s(t),\quad t>0,\label{2}\\
& & s(0)=0,\label{3}\\
& & \Psi_l(s(t),t)=\Psi_s(s(t),t)=0, \qquad t>0, \label{4}\\
& & k_s {\Psi_s}_x(s(t),t)-k_l {\Psi_l}_x (s(t),t)=\gamma s(t)^{\alpha} \dot s(t), \qquad t>0, \label{5}\\
& & k_l {\Psi_l}_x(0,t)=h_0 t^{-1/2} \left[ \Psi_l(0,t)-T_{\infty}t^{\alpha/2}\right]  \qquad t>0, \label{Convect}\\
& & \Psi_s(x,0)=-T_ix^{\alpha} , \qquad x>0 \label{tempInit}.
\end{eqnarray}
\noindent where the liquid (solid) phase is represented by the subscript $l$ ($s$), $\Psi$ is the temperature, $d$ is the diffusion coefficient,  $\gamma x^{\alpha}$ is the variable latent heat  per  unity of volume, $-T_ix^{\alpha}$ is the depth-varying initial temperature and the phase-transition temperature is zero. Condition (\ref{Convect}) represents the convective boundary condition at the fixed face $x=0$ . $T_{\infty}$ is the bulk temperature at a large distance from the fixed face $x=0$ and $h_0$ is the coefficient that characterizes the heat transfer at the fixed face. Moreover $\dot s(t)$ represents the velocity of the phase-change interface. We will work under the assumption  that $\gamma, T_i, T_{\infty}, h_0>0$ which corresponds to the melting case.

 In Section 2 we will quickly review fundamental results that allow us to apply the similarity transformation technique to our problem. We will analyse the fusion of a semi-infinite material which is initially at the solid phase, where  a convective condition is imposed at the fixed boundary $x=0$ and where the latent heat is considered as a power function of the position with power $\alpha$. In Section 3 we will provide an explicit solution of a similarity type of the problem (\ref{1})-(\ref{tempInit}) under certain conditions on the data, proving in addition its uniqueness in case that $\alpha$ is a positive non-integer exponent. We will study the particular case when $\alpha$ is a non-negative integer, recovering  for $\alpha=0$ the results obtained by \cite{Ta1}.  Finally  Section 4 will show that the solution to our problem converges to the solution of a different free boundary problem with a prescribed temperature at $x=0$ when  the coefficient $h_0\rightarrow +\infty$ which has been recently studied in \cite{ZhShZh2017}.

The main contribution of this paper is to generalize the work that has been done in  \cite{Ta1},\cite{ZhXi} and \cite{BoTa}, by obtaining the explicit solution of a one-dimensional two-phase Stefan problem for a semi-infinite material where a  variable latent heat and a convective boundary condition at the fixed face is considered; as well as to obtain the results given in \cite{ZhShZh2017} when the coefficient that characterizes the convective boundary condition goes to infinity.

\section{Explicit solution with latent heat depending on the position and a convective boundary condition at $x=0$.}

In this section it will be found the explicit solution of the problem governed by (\ref{1})-(\ref{tempInit}). The proof will be splitted into two subsections. The first one results from the work of Zhou and Xia in \cite{ZhXi} and  corresponds to the case when $\alpha$ is  positive and non-integer. The second one is correlated with the case when $\alpha$ is a non-negative integer, based on \cite{ZWB}.

\subsection{Case when $\alpha$ is a positive non-integer exponent.}

The following lemma has already been developed by Zhou-Xia in \cite{ZhXi} and  constitutes the base on which we will find solutions for the differential heat equations (\ref{1})-(\ref{2}).
\smallskip
\smallskip
\begin{lemma} \text{}  \label{Lemma2.1}

\begin{enumerate}
\item Let 
\begin{equation}\label{similarity}
\Psi(x,t)=t^{\alpha/2}f(\eta), \text{ with \quad} \eta=\dfrac{x}{2\sqrt{dt}}
\end{equation}

then $\Psi=\Psi(x,t)$ is a solution of the heat equation $\Psi_t(x,t)=d\Psi_{xx}(x,t)$, with $d>0$ if and only if $f=f(\eta)$ satisfies the following ordinary differential equation:

\begin{equation}\label{diffeq1}
\frac{d^2f}{d\eta^2}(\eta)+2\eta \frac{df}{d\eta}(\eta)-2\alpha f(\eta)=0.
\end{equation}

\item An equivalent formulation for equation (\ref{diffeq1}), introducing the new variable $z=-\eta^2$ is:
\begin{equation}\label{diffeq2}
z\frac{d^2f}{dz^2}(z)+\left(\frac{1}{2}-z \right)\frac{df}{dz}(z)+\frac{\alpha}{2}f(z)=0.
\end{equation}

\item The general solution of the ordinary differential equation (\ref{diffeq2}), called Kummer's equation, is given by:
\begin{equation}\label{genSol1}
f(z)=\widehat{c_{11}}M \left(-\dfrac{\alpha}{2},\dfrac{1}{2},z\right)+\widehat{c_{21}}U\left(-\dfrac{\alpha}{2},\dfrac{1}{2},z\right) .
\end{equation}

\noindent where $\widehat{c_{11}}$ and $\widehat{c_{21}}$ are arbitrary  constants and $M(a,b,z)$ and $U(a,b,z)$ are the Kummer functions defined by:

\begin{align}
& M(a,b,z)=\sum\limits_{s=0}^{\infty}\frac{(a)_s}{(b)_s s!}z^s , \text{ where b cannot be a  nonpositive integer,} \label{M} \\
& U(a,b,z)=\tfrac{\Gamma(1-b)}{\Gamma(a-b+1)}M(a,b,z)+\tfrac{\Gamma(b-1)}{\Gamma(a)} z^{1-b}M(a-b+1,2-b,z) \label{U}.
\end{align}
where $(a)_s$ is the pochhammer symbol:
\begin{equation}
 (a)_s=a(a+1)(a+2)\dots (a+s-1), \quad \quad (a)_0=1 
\end{equation}

\end{enumerate}
\end{lemma}
\begin{proof}
See \cite{ZhXi}.
\end{proof}

\begin{remark}
All the properties of Kummer's functions to be used in the following arguments can be found in the Appendix A.
\end{remark}

\begin{remark} \label{Remark1}
Taking into account  equation (\ref{genSol1}) and definition (\ref{U}) we can rewrite the general solution of the ordinary differential equation (\ref{diffeq2}) as:
\begin{equation}\label{genSol2}
f(z)=\overline{c_{11}}M \left(-\dfrac{\alpha}{2},\dfrac{1}{2},z\right)+\overline{c_{21}}z^{1/2} M\left(-\dfrac{\alpha}{2}+\dfrac{1}{2},\dfrac{3}{2},z\right),
\end{equation}
where $\overline{c_{11}}$ and $\overline{c_{21}}$ are arbitrary constants.
\end{remark}
                                                                                                                                                                                                                                                                    
\begin{remark} \label{Remark1-2}
Taking into account Lemma \ref{Lemma2.1} and Remark \ref{Remark1} we can assure that \mbox{$\Psi(x,t)=t^{\alpha/2}f(\eta)$} satisfies the heat equation $\Psi_t(x,t)=d\Psi_{xx}(x,t)$ if and only if it is defined as:
\begin{equation}\label{genSol2}
\Psi(x,t)=t^{\alpha/2}\left[ c_{11}M \left(-\dfrac{\alpha}{2},\dfrac{1}{2},-\eta^2\right)+c_{21}\eta M\left(-\dfrac{\alpha}{2}+\dfrac{1}{2},\dfrac{3}{2},-\eta^2\right) \right],
\end{equation}
with $\eta=\tfrac{x}{2\sqrt{dt}}$ and where $c_{11}$ and $c_{21}$ are arbitrary constants (not necessarily real).
\end{remark}

Our main outcome is given by the following theorem, which constitutes a generalization to the two-phase case of \cite{BoTa}. This theorem ensures the existence and uniqueness of solution of the problem (\ref{1})-(\ref{tempInit}) under a restriction for the convective coefficient, providing in addition the explicit solution.

\begin{thm}\label{TeoPcipal}
If the  coefficient $h_0$ satisfies the inequality:
\begin{equation}\label{DesH0}
h_0> \dfrac{2^{\alpha} \Gamma\left(\dfrac{\alpha}{2}+1\right) k_s T_id_s^{(\alpha-1)/2}}{ T_{\infty}  \sqrt{\pi}} 
\end{equation}
then there exists an instantaneous fusion process and the free boundary problem (\ref{1})-(\ref{tempInit}) has a unique solution of a similarity type given by:
\begin{align}
&s(t) =2\nu \sqrt{d_l t},\label{s}\\
&\Psi_l(x,t)=  t^{\alpha/2}\left[ E_l M\left(-\frac{\alpha}{2}, \frac{1}{2},-\eta_l^2\right)+F_l \eta_l M\left(-\frac{\alpha}{2}+\frac{1}{2},\frac{3}{2},-\eta_l^2\right)\right],\label{PsiLiq}\\
&\Psi_s(x,t)=t^{\alpha/2}\left[ E_s M\left(-\frac{\alpha}{2}, \frac{1}{2},-\eta_s^2\right)+F_s \eta_s M\left(-\frac{\alpha}{2}+\frac{1}{2},\frac{3}{2},-\eta_s^2\right)\right], \label{PsiSol} 
\end{align}
where $\eta_l=\tfrac{x}{2\sqrt{d_lt}}$, $\eta_s=\tfrac{x}{2\sqrt{d_st}}$  and the constants $E_l$, $F_l$, $E_s$ and $F_s$ are given by:
\begin{align}
& E_l=\dfrac{-\nu M\left(-\dfrac{\alpha}{2}+\dfrac{1}{2},\dfrac{3}{2},-\nu^2\right)}{M\left(-\dfrac{\alpha}{2},\dfrac{1}{2},-\nu^2\right)}F_l, \label{20} \\
&F_l=\dfrac{-h_0 T_{\infty}2 \sqrt{d_l}M\left(-\dfrac{\alpha}{2},\dfrac{1}{2},-\nu^2\right)}{\left[ k_lM\left(-\dfrac{\alpha}{2},\dfrac{1}{2},-\nu^2\right)+2\sqrt{d_l}h_0\nu M\left(-\dfrac{\alpha}{2}+\dfrac{1}{2},\dfrac{3}{2},-\nu^2\right)\right]},\label{21}\\
& E_s=\dfrac{-\nu \omega M\left(-\dfrac{\alpha}{2}+\dfrac{1}{2},\dfrac{3}{2},-\nu^2 \omega^2\right)}{M\left(-\dfrac{\alpha}{2},\dfrac{1}{2},-\nu^2\omega^2\right)}F_s,  \qquad \text{ with }\quad   \omega=\sqrt{d_l/d_s},\label{22}\\
& F_s=\dfrac{-T_i 2^{\alpha+1}d_s^{\alpha/2} M\left(\dfrac{\alpha}{2}+\dfrac{1}{2},\dfrac{1}{2},\nu^2\omega^2\right)}{U\left(\dfrac{\alpha}{2}+\dfrac{1}{2},\dfrac{1}{2},\nu^2\omega^2\right)}\label{23}.
\end{align}
and the dimensionless coefficient $\nu$ is the unique positive solution of the following equation:
\begin{eqnarray} \label{24}
-\dfrac{k_sT_id_s^{(\alpha-1)/2}}{\gamma d_l^{(\alpha+1)/2}}f_1(x)+\dfrac{h_0 T_{\infty}}{\gamma 2^{\alpha} d_l^{(\alpha+1)/2}}f_2(x)=x^{\alpha+1}, \qquad x>0.
\end{eqnarray}
in which functions $f_1$ and $f_2$ are defined by:
\begin{align}
&f_1(x)=\dfrac{1}{U\left(\dfrac{\alpha}{2}+\dfrac{1}{2},\dfrac{1}{2},x^2\omega^2\right)}, \qquad x>0,& \label{f1}\\
&f_2(x)=\dfrac{1}{\left[ M\left(\dfrac{\alpha}{2}+\dfrac{1}{2},\dfrac{1}{2},x^2\right)+2\dfrac{\sqrt{d_l}h_0}{k_l}x M\left(\dfrac{\alpha}{2}+1,\dfrac{3}{2},x^2\right)\right]}, \ x>0.&
\end{align}
\end{thm}

\begin{proof}
The general solution of equations (\ref{1})-(\ref{2}) based on Kummer functions are given by the Lemma \ref{Lemma2.1} and Remark \ref{Remark1-2}:
\begin{align}
& \Psi_l(x,t)=t^{\alpha/2}\left[E_l M\left(-\dfrac{\alpha}{2},\dfrac{1}{2},-\eta_l^2\right)+F_l \eta_l M\left(-\dfrac{\alpha}{2}+\dfrac{1}{2},\dfrac{3}{2},-\eta_l^2\right)  \right], \label{27}\\
& \Psi_s(x,t)=t^{\alpha/2}\left[E_s M\left(-\dfrac{\alpha}{2},\dfrac{1}{2},-\eta_s^2\right)+F_s \eta_s M\left(-\dfrac{\alpha}{2}+\dfrac{1}{2},\dfrac{3}{2},-\eta_s^2\right)  \right],
\end{align}
where $\eta_l=\dfrac{x}{2\sqrt{d_lt}}$, $\eta_s=\dfrac{x}{2\sqrt{d_st}}$, and $E_l$, $F_l$, $E_s$ and $F_s$ are coefficients that must be determined.

Furthermore, condition (\ref{4}) together with (\ref{27}) implies that the free boundary should take the following  form:
\begin{equation}
s(t)=2\nu \sqrt{d_l t}. 
\end{equation}
where $\nu$ is a constant that also has to be computed.

Using the derivation formulas for the Kummer functions (\ref{DerivM1})-(\ref{DerivM2}) presented in the Appendix A it is deduced that:
\begin{eqnarray}
{\Psi_l}_x (x,t)&=&\dfrac{t^{(\alpha-1)/2}}{\sqrt{d_l}} \left[ E_l \alpha \eta_l M\left(-\dfrac{\alpha}{2}+1,\dfrac{3}{2},-\eta_l^2 \right)+\right. \nonumber\\
& &+ \left. \dfrac{F_l}{2}M\left(-\dfrac{\alpha}{2}+\dfrac{1}{2},\dfrac{1}{2},-\eta_l^2 \right)\right], \label{DerivLiq}\\
{\Psi_s}_x (x,t)&=&\dfrac{t^{(\alpha-1)/2}}{\sqrt{d_s}} \left[ E_s \alpha \eta_s M\left(-\dfrac{\alpha}{2}+1,\dfrac{3}{2},-\eta_s^2 \right)+\right. \nonumber\\
&&+\left. \dfrac{F_s}{2}M\left(-\dfrac{\alpha}{2}+\dfrac{1}{2},\dfrac{1}{2},-\eta_s^2 \right)\right]. \label{DerivSol}
\end{eqnarray}

From  equation (\ref{4}) we have:
\begin{equation}
 t^{\alpha/2}\left[ E_lM\left(-\dfrac{\alpha}{2},\dfrac{1}{2},-\nu^2 \right)+F_l\nu M\left(-\dfrac{\alpha}{2}+\dfrac{1}{2},\dfrac{3}{2},-\nu^2 \right)\right]=0.
\end{equation}
Isolating $E_l$ we obtain (\ref{20}).

On the other hand, using (\ref{27}) and (\ref{DerivLiq}), condition (\ref{Convect}) becomes
\begin{equation}\label{34}
k_l \dfrac{F_l}{2\sqrt{d_l}}=h_0 \left[ E_l-T_{\infty}\right],
\end{equation}
and replacing $E_l$ given by (\ref{20}) into (\ref{34}) we get that $F_l$ is given by (\ref{21}).

Condition (\ref{4}), $\Psi_s(s(t),t)=0$ imply:

\begin{equation}
E_sM\left(-\dfrac{\alpha}{2},\dfrac{1}{2},-\nu^2\omega^2 \right) +F_s \nu\omega M\left( -\dfrac{\alpha}{2}+\dfrac{1}{2},\dfrac{3}{2},-\nu^2\omega^2\right)=0 \quad \text{ where } \omega=\sqrt{\tfrac{d_l}{d_s}},
\end{equation}
leading us to define $E_s$ by (\ref{22}).

In view of  condition (\ref{tempInit}), it is necessary to compute $\Psi_s(x,0)$, given by the expression:
\begin{eqnarray}
\Psi_s(x,0) &=& \lim\limits_{t\rightarrow 0} \Psi_s(x,t)=E_s \left[ \lim\limits_{t\rightarrow 0}  t^{\alpha/2} M\left(-\tfrac{\alpha}{2},\tfrac{1}{2},-\eta_s^2\right) \right] + \nonumber \\
&           & +F_s \left[ \lim\limits_{t\rightarrow 0}  t^{\alpha/2} \eta_s M\left(-\tfrac{\alpha}{2}+\tfrac{1}{2},\tfrac{3}{2},-\eta_s^2\right) \right] \label{Limit}
\end{eqnarray}

Taking into account formula (\ref{M&U}) from the Appendix A we obtain:
\begin{eqnarray}
M\left( -\dfrac{\alpha}{2},\dfrac{1}{2},-\eta_s^2\right)&=& \left[\dfrac{\sqrt{\pi}}{\Gamma\left( \frac{\alpha}{2}+\frac{1}{2}\right)}e^{-\frac{\alpha}{2}\pi i} U\left(-\dfrac{\alpha}{2},\dfrac{1}{2},-\eta_s^2 \right)+ \right. \nonumber \\
&+& \left. \dfrac{\sqrt{\pi}}{\Gamma\left(- \frac{\alpha}{2}\right)} e^{-\frac{(\alpha+1)}{2}\pi i } e^{-\eta_s^2} U\left(\dfrac{\alpha}{2}+\dfrac{1}{2},\dfrac{1}{2},\eta_s^2 \right)\right].\label{Rel1}
\end{eqnarray}

and
\begin{eqnarray}
&&M\left( -\dfrac{\alpha}{2}+\dfrac{1}{2},\dfrac{3}{2},-\eta_s^2\right)= \left[\dfrac{\sqrt{\pi}}{2\Gamma\left( \frac{\alpha}{2}+1\right)}e^{(-\frac{\alpha}{2}+\frac{1}{2})\pi i} U\left(-\dfrac{\alpha}{2}+\dfrac{1}{2},\dfrac{3}{2},-\eta_s^2 \right)+  \right. \nonumber \\
&&+ \left.\dfrac{\sqrt{\pi}}{2\Gamma\left(- \frac{\alpha}{2}+ \frac{1}{2}\right)} e^{-(\frac{\alpha}{2}+1)\pi i} e^{-\eta_s^2} U\left(\dfrac{\alpha}{2}+1,\dfrac{3}{2},\eta_s^2 \right)\right].\label{Rel2}
\end{eqnarray}

We can observe that if $\alpha$ is a non-negative even integer then $\Gamma\left(-\frac{\alpha}{2} \right)$  
is not defined, and so (\ref{Rel1}) is not valid. In the same way if $\alpha$ is a non-negative odd integer then $\Gamma\left(-\frac{\alpha}{2}+\frac{1}{2} \right)$ is neither defined and (\ref{Rel2}) cannot be applied.  From this fact we restrict $\alpha$ to be positive and non-integer. 

Considering (\ref{Rel1}) and (\ref{Rel2})and applying (\ref{HypGeom1}) we obtain the following limits:
\begin{eqnarray}
\lim\limits_{t\rightarrow 0}\left[  t^{\alpha/2}M\left(-\dfrac{\alpha}{2},\dfrac{1}{2},-\eta_s^2 \right)\right]&=&  \dfrac{\sqrt{\pi}}{\Gamma\left( \frac{\alpha}{2}+\frac{1}{2}\right)} \dfrac{x^{\alpha}}{2^{\alpha} d_s^{\alpha/2}}.\label{limit1}
\end{eqnarray} 
and
\begin{equation}
\lim\limits_{t\rightarrow 0 } t^{\alpha/2} \eta_s M\left(-\dfrac{\alpha}{2}+\dfrac{1}{2},\dfrac{3}{2},-\eta_s^2\right)= \dfrac{\sqrt{\pi}}{\Gamma\left(\frac{\alpha}{2}+1 \right)} \dfrac{x^{\alpha}}{2^{\alpha+1} d_s^{\alpha/2}}, \label{limit2}
\end{equation}
 Combining (\ref{Limit}), (\ref{limit1}) and (\ref{limit2}) we deduce that:
\begin{equation}\label{Psi_s(x,0)}
\Psi_s(x,0)= E_s \dfrac{\sqrt{\pi}}{\Gamma\left( \dfrac{\alpha}{2}+\dfrac{1}{2}\right)}\dfrac{x^{\alpha}}{(4d_s)^{\alpha/2}} +F_s \dfrac{\sqrt{\pi}}{2\Gamma\left(\dfrac{\alpha}{2}+1 \right)}\dfrac{x^{\alpha}}{(4d_s)^{\alpha/2}}
\end{equation}
Considering the initial temperature given by  (\ref{tempInit}), and replacing $E_s$ by (\ref{22}) in (\ref{Psi_s(x,0)}) it is obtained:
\begin{equation}
-\nu\omega \dfrac{M\left( -\dfrac{\alpha}{2}+\dfrac{1}{2},\dfrac{3}{2},-\nu^2\omega^2\right)}{M\left(-\dfrac{\alpha}{2},\dfrac{1}{2},-\nu^2\omega^2 \right)} \dfrac{\sqrt{\pi}}{\Gamma\left(\dfrac{\alpha}{2}+\dfrac{1}{2} \right)}F_s+ \dfrac{\sqrt{\pi}}{2\Gamma\left(\dfrac{\alpha}{2}+1 \right)}F_s=-T_i (4d_s)^{\alpha/2}
\end{equation}
Then we can determine $F_s$ using the definition of the $U$-Kummer function and the identity (\ref{MRelExp1}) presented in Appendix A arriving to definition (\ref{23}).

Until now we have obtained $E_l$, $F_l$, $E_s$ and $F_s$ as functions of $\nu$, arriving to the expressions (\ref{20})-(\ref{23}).

Finally it remains to take into account the Stefan condition (\ref{5}) from which we will deduce an equation that must be satisfied by the unknown coefficient $\nu$ that characterized the free boundary. Substituting equations (\ref{20})-(\ref{23}), (\ref{DerivLiq})-(\ref{DerivSol}) into (\ref{5}) and applying formula (\ref{MRelExp2}) it can be obtained that $\nu$ must satisfy the following equation:
\begin{eqnarray}
&& \dfrac{k_lh_0T_{\infty}}{\left[ k_l M\left(\dfrac{\alpha}{2}+\dfrac{1}{2},\dfrac{1}{2},x^2 \right)  + 2\sqrt{d_l}h_0 x M\left(\dfrac{\alpha}{2}+1,\dfrac{3}{2},x^2 \right)\right]}  + \qquad \nonumber\\
&&\qquad  -\dfrac{k_s T_i 2^{\alpha} d_s^{(\alpha-1)/2}}{U\left(\dfrac{\alpha}{2}+\dfrac{1}{2},\dfrac{1}{2},x^2\omega^2 \right)} =\gamma 2^{\alpha}x^{\alpha+1}d_l^{(\alpha+1)/2}, \qquad x>0. \label{EcuacionNu}
\end{eqnarray}
that can be rewritten, arriving to the result that $\nu$ must be a solution of the equation (\ref{24}).

Our proof is going to be completed by showing that there exists a unique solution $\nu$ for the equation (\ref{EcuacionNu}) (i.e. (\ref{24})). With this purpose we will study the behaviour of the functions $f_1$ and $f_2$.

On one hand, due to the derivation formula (\ref{DerivU}), and its integral representation (\ref{IntRepU}) we can assure that $f_1$ is an increasing function of $x$.  It follows immediately that the first term of the left hand side of equation (\ref{24}) decreases from $\Delta_1= -\dfrac{k_sT_i d_s^{(\alpha-1)/2}}{\gamma d_l^{(\alpha+1)/2}} \dfrac{\Gamma\left( \alpha/2+1\right)}{\sqrt{\pi}}$ to $-\infty$ when $x$ increases from $0$ to $+\infty$.

On the other hand, taking into account  equations (\ref{DerivM1}) and (\ref{DerivM2}) we arrive to the conclusion that $f_2$ is  a decreasing function of $x$. Therefore the second term of the left hand side of equation (\ref{24}) decreases from $\Delta_2=\dfrac{h_0T_{\infty}}{\gamma 2^{\alpha} d_l^{(\alpha+1)/2}}$ to $0$ when $x$ increases from $0$ to $+\infty$.

In consequence we can assure that the left hand side of (\ref{24}) decreases from $\Delta_1+ \Delta_2$ to $-\infty$ when $x$ increases from $0$ to $+\infty$.

As the right hand side of (\ref{24}) is an increasing function of $x$ that goes from $0$ to $+\infty$, we claim that the equation (\ref{24}) has a unique solution if and only if it is satisfied the following condition:
\begin{equation}
\Delta_1+\Delta_2 > 0
\end{equation} 
which is equivalent to (\ref{DesH0}).

\end{proof}

\begin{remark}
An inequality of the type (\ref{DesH0}) in order to obtain an instantaneous phase-change process was given firstly in \cite{Ta3}; see also \cite{Ro}.
\end{remark}

\begin{corollary} If the coefficient $h_0$ satisfies the following inequality:
\begin{equation}
0< h_0 \leq  \dfrac{2^{\alpha} \Gamma\left(\dfrac{\alpha}{2}+1\right) k_s T_id_s^{(\alpha-1)/2}}{ T_{\infty}\sqrt{\pi}} 
\end{equation}
then the free boundary problem (\ref{1})-(\ref{tempInit}) reduce to a classical heat transfer problem for the initial solid phase governed by:

\begin{eqnarray}
& &{\Psi_s}_t(x,t)=d_s {\Psi_s}_{xx}(x,t),\qquad x>0,\quad t>0,\\
& & k_s {\Psi_s}_x(0,t)=h_0 t^{-1/2} \left[ \Psi_s(0,t)-T_{\infty}t^{\alpha/2}\right],  \qquad t>0, \label{CondStefanSinFrontera}\\
& & \Psi_s(x,0)=-T_ix^{\alpha}, \qquad x>0, \label{CondTempInicialSinFrontera} 
\end{eqnarray}
whose explicit solution is given by:
\begin{eqnarray}
\Psi_s(x,t)=t^{\alpha/2}\left[E_s M\left(-\dfrac{\alpha}{2},\dfrac{1}{2},-\eta_s^2 \right) +F_s\eta_s M\left(-\dfrac{\alpha}{2}+\dfrac{1}{2},\dfrac{3}{2},-\eta_s^2 \right)\right],
\end{eqnarray}
where $\eta_s=x/\sqrt{4d_st}$ and :
\begin{eqnarray}
E_s &=&\dfrac{-T_i d_s^{\alpha/2}k_s \Gamma(\alpha+1)+\Gamma\left( \dfrac{\alpha+1}{2}\right)h_0 \sqrt{d_s}T_{\infty}}{\left[k_s \Gamma\left( \dfrac{\alpha}{2}+1\right)+h_0 \sqrt{d_s}\Gamma\left(\dfrac{\alpha+1}{2} \right) \right]}, \\
F_s&=&\dfrac{2 \sqrt{d_s}h_0 (E_s - T_{\infty})}{k_s}.
\end{eqnarray}

\end{corollary}

\begin{proof}
From Lemma \ref{Lemma2.1} and Remark \ref{Remark1-2} we have that the temperature is given by:
\begin{equation}
\Psi_s(x,t)=t^{\alpha/2} \left[E_s M\left(-\dfrac{\alpha}{2},\dfrac{1}{2}, -\eta_s^2 \right)+F_s \eta_s M\left( -\dfrac{\alpha}{2}+\dfrac{1}{2},\dfrac{3}{2}.-\eta_s^2\right) \right] 
\end{equation}
where $\eta_s=\dfrac{x}{\sqrt{4d_st}}$ and $E_s$ and $F_s$ are coefficients that must be determined.

Taking into account conditions (\ref{CondStefanSinFrontera})-(\ref{CondTempInicialSinFrontera}), coefficients $E_s$ and $F_s$  are obtained in an  analogous way as in the proof of the Theorem \ref{TeoPcipal}.

\end{proof}

\subsection{Case when $\alpha$ is a non-negative integer}

This section is intended to present the exact solution of the problem (\ref{1})-(\ref{tempInit}) in the particular case that $\alpha$ is a non-negative integer. Using formulas (\ref{AlphaInt1})-(\ref{AlphaInt2}) from Appendix A it can be  proved the following assertion.

\begin{lemma}
Consider the problem (\ref{1})-(\ref{tempInit}), where $\alpha=n\in \mathbb{N}_0$. If the coefficient $h_0$ satisfies the inequality:
\begin{equation}
h_0>\dfrac{2^n \Gamma\left(\dfrac{n}{2}+1 \right)k_s T_i d_s^{(n-1)/2}}{ T_{\infty}\sqrt{\pi}} \label{h0-Nat}
\end{equation}
then the explicit solution of this problem is given by:
\begin{align}
&s(t) =2\nu \sqrt{d_l t}, \label{Front-Nat}\\
&\Psi_l(x,t)= - \dfrac{t^{n/2}2^{n}h_0 T_{\infty} \sqrt{d_l} \Gamma\left(\dfrac{n}{2}+\dfrac{1}{2} \right) \Gamma\left(\dfrac{n}{2}+1 \right) \left[F_n(\eta_l)E_n(\nu)-F_n(\nu)E_n(\eta_l) \right]}{\left[k_l\Gamma\left(\dfrac{n}{2}+1 \right) E_n(\nu)+ \sqrt{d_l}h_0 \Gamma\left(\dfrac{n}{2}+\dfrac{1}{2} \right)F_n(\nu) \right]},\\
&\Psi_s(x,t)=t^{n/2} 2^n T_i d_s^{n/2}\Gamma(n+1) \left[ \dfrac{E_n(\eta_s)F_n(\nu\omega)-E_n(\nu\omega)F_n(\eta_s)}{E_n(\nu\omega)-F_n(\nu\omega)}\right],\label{TempSol-Nat}
\end{align}
where $\eta_l=\dfrac{x}{2\sqrt{d_lt}}$, $\eta_s=\dfrac{x}{2\sqrt{d_st}}$, $\omega=\sqrt{\dfrac{d_l}{d_s}}$ and $\nu$ is the unique solution of the following equation:
\begin{eqnarray}
&&\dfrac{h_0T_{\infty}}{\gamma 2^n d_l^{(n+1)/2}} \dfrac{1}{\left[e^{x^2} 2^n \Gamma\left(\dfrac{n}{2}+1 \right)E_n(x)+\dfrac{2^n \sqrt{d_l} h_0}{k_l}e^{x^2}\Gamma\left( \dfrac{n}{2}+\dfrac{1}{2}\right)F_n(x)\right]}+  \nonumber \\
&& -\dfrac{k_sT_i d_s^{(n-1)/2}}{\gamma d_l^{(n+1)/2}} \dfrac{1}{2^n e^{x^2\omega^2}\sqrt{\pi} \left(E_n(x\omega)-F_n(x\omega) \right)}  = x^{n+1}, \quad\quad  x>0 . \label{EcNu-Nat}
\end{eqnarray}
\end{lemma}

\begin{proof}
Inequality (\ref{h0-Nat}), functions (\ref{Front-Nat})-(\ref{TempSol-Nat}) and equation (\ref{EcNu-Nat}) can be deduced following the same reasoning used in the demonstration of the Theorem \ref{TeoPcipal}  by using  the relationship between the Kummer functions and the family of the repeated integrals of the complementary error function given by (\ref{AlphaInt1}) and (\ref{AlphaInt2}).

Let us note that in order to follow the arguments of Theorem \ref{TeoPcipal} we must show that the limits given by (\ref{limit1}) and (\ref{limit2}) remain true in case that $\alpha$ is a non-negative integer. But this can be easily proved due to the formula presented by Tao in \cite{Tao}:
\begin{equation}
\lim\limits_{t\rightarrow 0}\text{ } t^{n/2} E_n\left(\eta_s \right)=\lim\limits_{t\rightarrow 0}\text{ } t^{n/2} F_n\left(\eta_s \right)= \dfrac{x^n}{\Gamma\left( n+1\right) 2^n d_s^{n/2}}.
\end{equation}
and due to the Legendre duplication formula for the Gamma function \cite{Ab}:
\begin{equation}
\Gamma(x)\Gamma\left(x+\frac{1}{2}\right)=\dfrac{\sqrt{\pi}}{2^{2x-1}}\Gamma(2x).
\end{equation}

\end{proof}

\begin{remark} \label{ObsN=0}
Considering $n=0$ and taking into account that $E_0(z)=1$ and $F_0(z)=erf(z)$, condition (\ref{h0-Nat}) and functions (\ref{Front-Nat})-(\ref{TempSol-Nat}) reduce to: 
\begin{align}
& h_0>\dfrac{k_s T_i}{T_{\infty}\sqrt{\pi d_s}}\\
&s(t) =2\nu \sqrt{d_l t},\\
&\Psi_l(x,t)= \dfrac{h_0 T_{\infty}\sqrt{\pi d_l}}{k_l}\dfrac{\left[erf(\nu)-erf\left( \dfrac{x}{2\sqrt{d_l t}}\right) \right]}{\left[1+\dfrac{\sqrt{\pi d_l} h_0}{k_l}erf(\nu) \right]},\\
&\Psi_s(x,t)= -T_i\left[ 1-\dfrac{erfc\left( \dfrac{x}{2\sqrt{d_st}}\right)}{erf(\nu \omega)}\right],
\end{align}
where $\nu$ is the unique solution of the following equation:
\begin{equation}
-\dfrac{k_s T_i}{\gamma\sqrt{\pi d_l d_s}} \dfrac{e^{-x^2\omega^2}}{erfc(x\omega)}+\dfrac{h_0T_{\infty}}{\gamma \sqrt{d_l}} \dfrac{e^{-x^2}}{\left[1+\dfrac{\sqrt{\pi d_l} h_0 erf(x)}{k_l} \right]}=x, \quad\quad x>0.
\end{equation}

This formulas are in agreement with the explicit solution of the problem presented by Tarzia in \cite{Ta1} which in contrast with our problem corresponds to a solidification process. 
\end{remark}

\begin{remark}
The results of Remark \ref{ObsN=0} in the one-phase case with a convective boundary condition are also recovered  in \cite{BoTa}.
\end{remark}

\section{Limit behaviour when $h_0\rightarrow +\infty$}

In this section we are going to study the limit behaviour of the solution of the problem governed by equation (\ref{1})-(\ref{tempInit}) when the coefficient $h_0$ that characterizes the heat transfer in the convective condition (\ref{Convect}) tends to infinity.  The main reason for doing this analysis is due to the fact that the convective heat input:
\begin{equation}
k_l{\Psi_{l}}_x(0,t)=h_0 t^{-1/2}[\Psi_l(0,t)-T_{\infty}t^{\alpha/2}], \label{Convectiva}
\end{equation}
 constitutes a generalization of the Dirichlet condition in the sense that if we take the limit when $h_0\rightarrow\infty$ in (\ref{Convectiva}) we must obtain $\Psi_l(0,t)=T_{\infty}t^{\alpha/2}$. Therefore, we will prove that the solution to our problem in which we consider a convective condition at the fixed face $x=0$,  converges to the solution of a problem with a temperature condition at the fixed face.

Bearing in mind that the solution to the problem (\ref{1})-(\ref{tempInit}), it means the free boundary and the temperatures in the solid and the liquid phase depends on $h_0$, we will rename them as: \newline

$
\left\lbrace
\begin{array}{lcl}
s_{h_0}(t) &:& \text{free boundary given by (\ref{s}), } \\
\nu_{h_0} &:& \text{unique solution of the equation (\ref{24}), } \\
\Psi_{l h_0}(t) &:& \text{liquid temperature given by (\ref{PsiLiq}), } \\
\Psi_{s h_0}(t) &:& \text{liquid temperature given by (\ref{PsiSol}). } 
\end{array}
\right.
$
\begin{thm}
Let us consider the problem given by conditions (\ref{1})-(\ref{tempInit}), where the solution  $s_{h_0}$, ${\Psi_l}_{h_0}$, ${\Psi_s}_{h_0}$ and $\nu_{h_0}$   is defined by (\ref{s}), (\ref{PsiLiq}), (\ref{PsiSol}) and (\ref{24}) respectively. If we take the limit when $h_0\rightarrow\infty$ we obtain that  $s_{h_0}$, ${\Psi_l}_{h_0}$, ${\Psi_s}_{h_0}$ and $\nu_{h_0}$  converge to  $s_{\infty}$, ${\Psi_{l\infty}}$, ${\Psi_{s\infty}}$ and $\nu_{\infty}$ respectively,  which corresponds to the solution of the following problem:
\begin{align}
& {\Psi_{l \infty }}_t(x,t)=d_l {\Psi_{l \infty }}_{xx}(x,t), \qquad 0<x<s_{\infty}(t), \quad t>0,& \label{1Infty}\\
& {\Psi_{s \infty }}_t(x,t)=d_s {\Psi_{s \infty }}_{xx}(x,t),\qquad x>s_{\infty}(t),\quad t>0,\label{2Infty}&\\
&  s_{\infty}(0)=0,\label{3Infty}&\\
&  \Psi_{l \infty }(s_{\infty}(t),t)=\Psi_{s \infty }(s_{\infty}(t),t)=0, \qquad t>0, \label{4Infty}&\\
& k_s {\Psi_{s \infty }}_x(s_{\infty}(t),t)-k_l {\Psi_{l \infty }}_x (s_{\infty}(t),t)=\gamma s_{\infty}(t)^{\alpha} \dot s_{\infty}(t), \qquad t>0,& \label{5Infty}\\
&  \Psi_{l \infty }(0,t)=T_{\infty} t^{\alpha/2}  \quad t>0, \label{TempInfty}&\\
 & \Psi_{s \infty }(x,0)=-T_ix^{\alpha} , \qquad x>0& \label{tempInitInfty}.
\end{align}
with $s_\infty=2\nu_{\infty }\sqrt{d_l t}$ and where a temperature $T_{\infty} t^{\alpha/2}$ is prescribed at the fixed face $x=0$. 
\end{thm}

\begin{proof}
On one hand if we consider the problem governed by the equations (\ref{1Infty})-(\ref{tempInitInfty}) we can obtain by following similar arguments of the proof of Theorem \ref{TeoPcipal} that the solution is given by:
\begin{eqnarray}
&&s_{\infty}(t)= 2\nu_{\infty}\sqrt{d_l t}, \\
&&\Psi_{l\infty}(x,t)= t^{\alpha/2}\left[ E_{l\infty} M\left( -\tfrac{\alpha}{2},\tfrac{1}{2},-\nu_{\infty}^2\right) +F_{l\infty} \nu_{\infty} M\left( -\tfrac{\alpha}{2}+\tfrac{1}{2},\tfrac{3}{2},-\nu_{\infty}^2\right)\right],\qquad \   \\
&&\Psi_{s\infty}(x,t)= t^{\alpha/2}\left[ E_{s\infty} M\left( -\tfrac{\alpha}{2},\tfrac{1}{2},-\nu_{\infty}^2\right) +F_{s\infty} \nu_{\infty} M\left( -\tfrac{\alpha}{2}+\tfrac{1}{2},\tfrac{3}{2},-\nu_{\infty}^2\right)\right], 
\end{eqnarray}
where
\begin{eqnarray}
E_{l \infty}&=& T_{\infty}, \\
F_{l \infty}&=& -\dfrac{T_{\infty} M\left( -\dfrac{\alpha}{2},\dfrac{1}{2},-\nu_{\infty}^2\right)}{\nu_{\infty} M\left( -\dfrac{\alpha}{2}+\dfrac{1}{2},\dfrac{3}{2},-\nu_{\infty}^2\right)},\quad \\
E_{s \infty}&=& -\dfrac{\nu_{\infty} \omega M\left( -\dfrac{\alpha}{2}+\dfrac{1}{2},\dfrac{3}{2},-\nu_{\infty}^2 \omega^2\right)}{M\left( -\dfrac{\alpha}{2},\dfrac{1}{2},-\nu_{\infty}^2\omega^2\right)}F_{s\infty},\\
F_{s \infty}&=& -\dfrac{T_i 2^{\alpha+1} d_s^{\alpha/2} M\left(\dfrac{\alpha}{2}+\dfrac{1}{2},\dfrac{1}{2},\nu_{\infty}^2\omega^2 \right)}{U\left(\dfrac{\alpha}{2}+\dfrac{1}{2},\dfrac{1}{2},\nu_{\infty}^2 \omega^2\right)},
\end{eqnarray}
with $\omega=\sqrt{\dfrac{d_l}{d_s}}$ and where $\nu_{\infty}$ is the unique solution of equation:
\begin{equation}\label{NuInfTemp}
\dfrac{k_l T_{\infty}}{2^{\alpha+1} d_l^{(\alpha/2+1)}\gamma} f_3(x)-\dfrac{k_s T_i d_s^{(\alpha-1)/2}}{\gamma  d_l^{(\alpha+1)/2}}f_1(x)=x^{\alpha+1}, \qquad x>0, 
\end{equation}
with
\begin{equation}
f_3(x)=\dfrac{1}{xM\left( \dfrac{\alpha}{2}+1,\dfrac{3}{2},x^2\right)} , \qquad x>0
\end{equation}
and $f_1(x)$ defined in (\ref{f1}).

The proof that $\nu_{\infty}$ is the unique solution of (\ref{NuInfTemp}) derive from analysing the growth of functions $f_1$ and $f_3$. On one hand, we have seen in the proof of Theorem (\ref{TeoPcipal}) that $f_1$ is an increasing function that satisfies $f_1(0)=\tfrac{\Gamma\left(\alpha/2+1 \right)}{\sqrt{\pi}}$ and $f_1(+\infty)=+\infty$. On the other hand, taking into account the derivation formula (\ref{DerivM1}) we can easily prove that $f_3(x)$ is a decreasing function that verifies $f_3(0)=+\infty$ and $f_3(+\infty)=0$. Thus we obtain that the left hand side of equation (\ref{NuInfTemp}) is a decreasing function that goes from $+\infty$ to $-\infty$ when $x$ goes from 0 to $+\infty$. As the right hand side of equation (\ref{NuInfTemp}) is an increasing function that increases from 0 to $+\infty$, we can assure that (\ref{NuInfTemp}) has a unique positive solution.  We remark here that the solution of the problem (\ref{1Infty})-(\ref{tempInitInfty}) was obtained in \cite{ZhShZh2017} by using results for a heat flux condition from an argument not so clear for us, and for this reason we have proved it with details.

Once we have calculated the solution of the problem (\ref{1Infty})-(\ref{tempInitInfty}), let us show that the solution of the problem (\ref{1})-(\ref{tempInit}) converges to it when $h_0\rightarrow +\infty$.
We know that $\nu_{h_0}$, which is the parameter that characterizes the free front in (\ref{1})-(\ref{tempInit}), is the unique solution of (\ref{24}). Taking limit in (\ref{24}) we obtain:

\begin{eqnarray}
&& \lim\limits_{h_0 \rightarrow +\infty}  \left[ -\tfrac{k_s T_i d_s^{(\alpha-1)/2}}{\gamma d_l^{(\alpha+1)/2}} \tfrac{1}{U\left( \tfrac{\alpha}{2}+\tfrac{1}{2},\tfrac{1}{2},x^2\omega^2\right)} \right] + \nonumber \\
&&+\lim\limits_{h_0 \rightarrow +\infty}\left[  \tfrac{h_0 T_{\infty}}{\gamma 2^{\alpha} d_l^{(\alpha+1)/2}} \tfrac{1}{\left[ M\left(\tfrac{\alpha}{2}+\tfrac{1}{2},\tfrac{1}{2},x^2\right)+2\tfrac{\sqrt{d_l}h_0}{k_l}x M\left(\tfrac{\alpha}{2}+1,\tfrac{3}{2},x^2\right)\right]} \right]= \nonumber \\ 
&&= -\dfrac{k_s T_i d_s^{(\alpha-1)/2}}{\gamma d_l^{(\alpha+1)/2}} \tfrac{1}{U\left( \tfrac{\alpha}{2}+\tfrac{1}{2},\tfrac{1}{2},x^2\omega^2\right)}+\tfrac{k_l T_{\infty}}{\gamma 2^{\alpha+1} d_l^{(\alpha/2+1)}} \tfrac{1}{x M\left(\tfrac{\alpha}{2}+1,\tfrac{3}{2},x^2\right)}= \nonumber \\
&&= -\dfrac{k_s T_i d_s^{(\alpha-1)/2}}{\gamma d_l^{(\alpha+1)/2}} f_1(x)+ \dfrac{k_l T_{\infty}}{\gamma 2^{\alpha+1} d_l^{(\alpha/2+1)}} f_3(x). \label{EcuacionNuInfinito}
\end{eqnarray}

That means that $\lim\limits_{h_0\rightarrow +\infty} \nu_{h_0}$ must be a solution of equation  (\ref{NuInfTemp}) which has a unique solution $\nu_{\infty}$, so we can conclude that $\lim\limits_{h_0\rightarrow +\infty} \nu_{h_0}=\nu_{\infty}$.

Subsequently  by simple algebraic calculations we obtain:
\begin{eqnarray}
&& \lim\limits_{h_0\rightarrow +\infty} s_{h_0}(t)= s_{\infty}(t), \\
&& \lim\limits_{h_0\rightarrow +\infty} \Psi_{l h_0} (x,t)= \Psi_{l\infty}(x,t), \\
&& \lim\limits_{h_0\rightarrow +\infty} \Psi_{s h_0} (x,t)= \Psi_{s\infty}(x,t). 
\end{eqnarray}

\end{proof}

\section{Conclusions} 
 In this article a closed analytical solution of a similarity type have been obtained for a one-dimensional two-phase Stefan problem in a semi-infinite material  using Kummer functions. The novel feature in the problem studied concerns a variable latent heat that depends on the position as well as a convective boundary condition at the fixed face $x=0$ of the material. Assuming a latent heat defined as a power function of the position allows the generalization of some previous theoretical results.  We have also generalized the classical two-phase Stefan problem with constant latent heat and a convective boundary condition  \cite{Ta1}  and the one-phase Stefan problem with latent heat depending on the position and a convective boundary  condition at the fixed face $x=0$ \cite{BoTa}.

Furthermore, we have shown that when $h_0$ increases, the solution of the problem  (\ref{1})-(\ref{tempInit}) converges  to the solution of a different free boundary problem (\ref{1Infty})-(\ref{tempInitInfty}) where a temperature condition at the fixed face is considered instead of a convective one \cite{ZhShZh2017}.

The key contribution of this paper has been to prove the existence and uniqueness of the explicit solution of the problem (\ref{1})-(\ref{tempInit}) when a restriction on the data is satisfied. 
We have presented the exact solution which is worth finding not only  to understand better the process involved but also to verify the accuracy of numerical methods that solve Stefan problems.

\appendix
\numberwithin{equation}{section}

\section{ }
% \section*{Appendix A}
This appendix presents a review of some of the significant mathematical  results of the Kummer functions which are used in the main body of the paper.

\vspace{0.3cm} 
\noindent \textbf{Definition of Kummer functions}
\vspace{0.3cm}

\noindent Kummer functions are defined by:

\begin{eqnarray}
 M(a,b,z)&=&\sum\limits_{s=0}^{\infty}\frac{(a)_s}{(b)_s s!}z^s ,\qquad \text{ with }b \text{ nonpositive integer,} \label{DefM}\\
 U(a,b,z)&=&\frac{\Gamma(1-b)}{\Gamma(a-b+1)}M(a,b,z)+\frac{\Gamma(b-1)}{\Gamma(a)} z^{1-b}M(a-b+1,2-b,z) \label{DefU}
\end{eqnarray}
where $(a)_s$ is the pochhammer symbol:
\begin{equation}
(a)_s=a(a+1)(a+2)\dots (a+s-1), \quad \quad (a)_0=1
\end{equation}
and  $\Gamma(\cdot)$ is the Gamma function. In order that $U$ is well-defined it is necessary that $a$ and $a-b+1$ be non-positive integers.

\vspace{0.4cm}
\noindent\textbf{Differentiation formulas}
\vspace{0.3cm}

 \noindent From \cite{OLBC} we have :

\begingroup
\addtolength{\jot}{0.2em}
\begin{align}
& \dfrac{d}{dz}M(a,b,z) = \dfrac{a}{b}M(a+1,b+1,z) \label{DerivM1}\\
&\dfrac{d}{dz}\left(z^{b-1}M(a,b,z) \right)=(b-1) z^{b-2}M(a,b-1,z) \label{DerivM2} \\
& \dfrac{d}{dz}U(a,b,z)=-aU(a+1,b+1,z) \label{DerivU}
\end{align}
\endgroup

\smallskip
\noindent\textbf{Connection Formulas}
\smallskip

 \noindent From \cite{OLBC} and \cite{ZhXi}   we know that :

\begin{itemize}
\item \textbf{Relationship with the generalized hypergeometric function:}
\begingroup
\addtolength{\jot}{0.2em}
\begin{align}
&  U(a,b,z)  \sim z^{-a}, \quad z\rightarrow \infty, \vert z \vert\leq \dfrac{3}{2}\pi -\delta \quad\text{ where } \delta \text{ is an arbitrary small positive constant}. \label{HypGeom1}
%& \lim_{z\rightarrow +\infty} \text{}_2F_0(0,p-q+1,-,-z^{-1})=1 \label{HypGeom2} 
\end{align}
\endgroup

\smallskip
\item \textbf{Integral Representation of $U$:}
\begin{equation}
U(a,b,z)=\dfrac{1}{\Gamma(a)}\int\limits_0^{\infty} e^{-zt}t^{a-1}(1+t)^{b-a-1}dt \qquad \text{ with } \text{Re}(a)>0 \text{ and } \vert \text{ph} (z)\vert<\dfrac{\pi}{2} \label{IntRepU}
\end{equation}

\smallskip
\item \textbf{Relationship between $U$ and $M$:}
\begin{equation}
 \dfrac{1}{\Gamma(b)}M(a,b,z)=\dfrac{e^{a\pi i}}{\Gamma(b-a)}U(a,b,z)+\dfrac{e^{-(b-a)\pi i}}{\Gamma(a)}e^z U(b-a,b,e^{-\pi i}z) \label{M&U}
\end{equation}

\smallskip
\item \textbf{Relationship with the exponential function:}
\begin{align}
&  M(a,b,z)= e^z M(b-a,b,-z) \qquad \qquad \qquad \qquad \label{MRelExp1}& \\
&e^{-z^2}= -2\alpha z^2M\left(-\dfrac{\alpha}{2}+\dfrac{1}{2},\dfrac{3}{2},-z^2 \right)M\left( -\dfrac{\alpha}{2}+1,\dfrac{3}{2},-z^2\right)+\nonumber &\\
&\quad\quad \ +M\left( -\dfrac{\alpha}{2},\dfrac{1}{2},-z^2\right)M\left( -\dfrac{\alpha}{2}+\dfrac{1}{2},\dfrac{1}{2},-z^2\right)& \label{MRelExp2}
\end{align}
where $\alpha$ is real and non-negative.

\smallskip
\item \textbf{Relationship with the family of the repeated integrals of the complementary error function:}
\begingroup
\addtolength{\jot}{0.2em}
\begin{eqnarray}
 M\left(-\dfrac{n}{2},\dfrac{1}{2},-z^2 \right)  &=&2^n \Gamma\left(\dfrac{n}{2}+1 \right)E_n(z) \label{AlphaInt1}\\
 zM\left(-\dfrac{n}{2}+\dfrac{1}{2},\dfrac{3}{2},-z^2 \right)&=&2^{n-1}\Gamma\left( \dfrac{n}{2}+\dfrac{1}{2}\right)F_n(z) \label{AlphaInt2}
\end{eqnarray}
\endgroup
where $n$ is an integer, $E_n$ and $F_n$ are defined by:
\begingroup
\addtolength{\jot}{0.2em}
\begin{align}
&  E_n(z)=\left[i^n erfc(z)+i^nerfc(-z) \right]/2 \\
& F_n(z)=\left[i^n erfc(-z)+i^n erfc(z) \right]/2
\end{align}
\endgroup
in which:

\begingroup
\addtolength{\jot}{0.1em}
\begin{align}
& i^0 erfc(x)=erfc(x) \\
& i^n erfc(x)=\int\limits_{x}^{+\infty} i^{n-1}erfc(t)dt 
\end{align}
\endgroup

\end{itemize}

\begin{scriptsize}
\begin{tabular}{lll}

& &\\
{\large \textbf{Nomenclature}}&   &\\
& &\\
$d_l, d_s$ &  & Diffusivity coefficient, $[m^2/s]$.\\
$h_0$ & & Coefficient that characterizes the heat transfer in condition (\ref{Convect}), $[kg/(^{\circ}Cs^{5/2})]$.\\
$k_l,k_s$ & & Thermal conductivity, $[W/(m ^{\circ}C)]$. \\
$s$ &  & Position of the free front, $[m]$.\\
$t$& & Time, $[s]$.\\
$T_{\infty}$ & & Coefficient that characterizes the bulk temperature in condition (\ref{Convect}), $[^{\circ}C/s^{\alpha/2}]$.\\
$T_i$ & & Coefficient that characterizes the initial temperature of the material in condition (\ref{tempInit}), $[^{\circ}C/m^{\alpha}]$.\\
$x$ & & Spatial coordinate, $[m]$.\\
& &\\
\mbox{Greek symbols } &   &\\
%& &\\
$\alpha$ & & Power of the position that characterizes the latent heat per unit volume, dimensionless.\\
$\gamma$ & & Coefficient that characterizes the latent heat per unit volume, $[kg/(s^2m^{\alpha+1})]$.\\
$\nu$ & & Coefficient that characterizes the free interface, dimensionless.\\
$\eta$ & & Similarity variable in expression (\ref{similarity}), dimensionless.\\
$\Psi_l,\Psi_s$ & & Temperature, $[^{\circ}C]$.\\
& &\\
Subscripts &   &\\
%& &\\
$l$ & & liquid phase.\\
$s$ & & solid phase.\\
& &\\
\end{tabular}
\end{scriptsize}

\section*{Acknowledgements}

The present work has been partially sponsored by the Projects PIP No 0534 from CONICET-UA and ANPCyT PICTO Austral 2016 No 0090, Rosario, Argentina.

\end{document}